\documentclass{article}
\usepackage[utf8]{inputenc}
\usepackage{fullpage}
\usepackage{amsthm}
\usepackage{amsmath}
\usepackage{amssymb}
\usepackage{mathtools}
\usepackage[colorlinks,citecolor=blue,bookmarks=true,linktocpage]{hyperref}
\usepackage{aliascnt}
\usepackage[numbered]{bookmark}
\usepackage[capitalise]{cleveref}
\usepackage[backend=biber,url=false,arxiv=abs,maxbibnames=100,isbn=false,sortcites=true]{biblatex}
\usepackage{cleveref}
\usepackage{xcolor}

\usepackage[T1]{fontenc}

\newtheorem{theorem}{Theorem}
\newtheorem{lemma}{Lemma}
\newtheorem{corollary}{Corollary}
\newtheorem{observation}{Observation}

\newcommand{\Next}{\vartriangleleft}

\newcommand{\prb}[1]{\textnormal{\scshape #1}}

\title{Finding a reconfiguration sequence between longest increasing subsequences\thanks{This work was partially supported by JSPS Kakenhi Grant Numbers
JP20H00595, 
JP21K11752, 
JP22H00513, and 
JP23H03344. 
}}
\author{Yuuki Aoike \and
Masashi Kiyomi \and
Yasuaki Kobayashi \and
Yota Otachi}

\begin{document}

\maketitle

\begin{abstract}
    In this note, we consider the problem of finding a step-by-step transformation between two longest increasing subsequences in a sequence, namely \prb{Longest Increasing Subsequence Reconfiguration}.
    We give a polynomial-time algorithm for deciding whether there is a reconfiguration sequence between two longest increasing subsequences in a sequence.
    This implies that \prb{Independent Set Reconfiguration} and \prb{Token Sliding} are polynomial-time solvable on permutation graphs, provided that the input two independent sets are largest among all independent sets in the input graph.
    We also consider a special case, where the underlying permutation graph of an input sequence is bipartite.
    In this case, we give a polynomial-time algorithm for finding a shortest reconfiguration sequence (if it exists).
\end{abstract}

\section{Introduction}
For a nonnegative integer $n$, we define $[n] = \{1, 2, \ldots, n\}$.
Let $A = (a_i)_{i = 1, 2, \ldots, n}$ be a sequence of distinct integers between $1$ and $n$.
We say that $I \subseteq [n]$ is \emph{feasible} (for $A$) if $a_i < a_j$ for $i, j \in I$ with $i < j$.
In other words, $I$ is the set of indices of an increasing subsequence of $A$.
A \emph{maximum feasible set} (for $A$) is a feasible set $I$ for $A$ such that there is no feasible set (for $A$) with cardinality strictly larger than $I$.
The problem of computing a maximum feasible set of a given sequence $A$, also known as \prb{Longest Increasing Subsequence}, is a typical example that can be solved in polynomial time with dynamic programming~\cite{cormen22}.

In this note, we consider the reconfiguration-variant of \prb{Longest Increasing Subsequence}, defined as follows.
Given a sequence of $n$ distinct integers $A$ and (not necessarily maximum) two feasible sets $I$ and $J$ with $|I| = |J|$, the goal is to determine whether there is a sequence of feasible sets $I_0, I_1, \ldots, I_\ell$ such that $I_0 = I$, $I_\ell = J$, and for $1 \le i \le \ell$, $I_{i}$ is obtained from $I_{i-1}$ by simultaneously adding element $j \notin I_{i-1}$ and removing $k \in I_{i-1}$ (i.e., $I_{i} = (I_{i-1} \cup \{j\}) \setminus \{k\}$).
We call this problem \prb{Increasing Subsequence Reconfiguration} and such a sequence \emph{a reconfiguration sequence between $I$ and $J$}.
If two input sets are maximum feasible sets for $A$, we particularly call the problem \prb{Longest Increasing Subsequence Reconfiguration}.
In this paper, we give a polynomial-time algorithm for \prb{Longest Increasing Subsequence Reconfiguration}.

\begin{theorem}\label{thm:main}
    \prb{Longest Increasing Subsequence Reconfiguration} can be solved in polynomial time.
\end{theorem}

\prb{Increasing Subsequence Reconfiguration} can be seen as a special case of a well-studied reconfiguration problem, called \prb{Independent Set Reconfiguration}.
Given a graph $G = (V, E)$ and two independent sets $I, J$ of $G$ with $|I| = |J|$, \prb{Independent Set Reconfiguration} asks whether there is a sequence of independent sets $I_0, I_1, \ldots, I_\ell$ such that $I_0 = I$, $I_\ell = J$, and for $1 \le i \le \ell$, $I_{i} \setminus I_{i - 1} = \{v\}$ and $I_{i - 1} \setminus I_{i} = \{u\}$ for some $u, v \in V$.
\prb{Increasing Subsequence Reconfiguration} corresponds to \prb{Independent Set Reconfiguration} on permutations graphs: An undirected graph $G = (V, E)$ with $V = [n]$ is called a \emph{permutation graph} if there is a permutation $\pi\colon [n] \to [n]$ such that for $1 \le i < j \le n$, $\pi(i) > \pi(j)$ if and only if $\{i, j\} \in E$.
Observe that for $I \subseteq V$, $I$ is an independent set of the permutation graph $G$ if and only if $I$ is a feasible set for $A = (\pi(i))_{i = 1, 2, \ldots, n}$.
Thus, our problem, \prb{Increasing Subsequence Reconfiguration}, is equivalent to \prb{Independent Set Reconfiguration} on permutation graphs.
\prb{Token Sliding} is a variant of \prb{Independent Set Reconfiguration}, where two vertices $u, v$ in the above definition are required to be adjacent in $G$.
It is easy to see that if $I$ and $J$ are maximum interdependent sets of $G$, these two problems are equivalent.

\begin{corollary}
    \prb{Independent Set Reconfiguration} and \prb{Token Sliding} can be solved in polynomial time, provided that the input graph $G$ is a permutation graph and two sets $I$ and $J$ are maximum independent sets of $G$.
\end{corollary}

This graph-theoretic perspective of \prb{Longest Increasing Subsequence Reconfiguration} gives another interesting consequence of finding a \emph{shortest} reconfiguration sequence between maximum independent sets on bipartite permutation graphs.
For any reconfiguration sequence $(I_0, I_1, \ldots, I_\ell)$, we have $\ell \ge |I_0 \setminus I_\ell|$, as we can remove at most one element from $I_0 \setminus I_\ell$ in a single step.
For bipartite permutation graphs, we can always find a reconfiguration sequence between maximum independent sets $I$ and $J$ with length $\ell = |I \setminus J|$ if there is a reconfiguration sequence between them.

\begin{theorem}\label{thm:bp}
    Let $G$ be a bipartite permutation graph and let $I$ and $J$ be maximum independent sets of $G$.
    Suppose that there is a reconfiguration sequence between $I$ and $J$.
    Then, there is a reconfiguration sequence of length $|I \setminus J|$ between $I$ and $J$.
\end{theorem}

The proof of \Cref{thm:bp} implies a polynomial-time algorithm for the ``shortest-sequence variant'' of \prb{Longest Increasing Subsequence Reconfiguration} when the underlying permutation graph of an input sequence $A$ is restricted to be bipartite.

\paragraph*{Related work} \prb{Independent Set Reconfiguration} and \prb{Token Sliding} are both known to be PSPACE-complete~\cite{BelmonteKLMOS21,BonsmaC:TCS:Finding:2009,HearnD05,KaminskiMM12} and studied on many graph classes.
\prb{Independent Set Reconfiguration} is solvable in polynomial time on even-hole free graphs~\cite{KaminskiMM12} and cographs~\cite{BonamyB14a,Bonsma16}, while it is NP-complete on bipartite graphs~\cite{LokshtanovM19}.
\prb{Token Sliding} is solvable in polynomial time on cographs~\cite{KaminskiMM12}, bipartite permutation graphs~\cite{Fox-EpsteinHOU15}, and interval graphs~\cite{BonamyB17,BrianskiFHM21}, while it is PSPACE-complete on split graphs~\cite{BelmonteKLMOS21} and bipartite graphs~\cite{LokshtanovM19}.
The result of \cite{Fox-EpsteinHOU15} does not yield \Cref{thm:bp} since their polynomial-time algorithm may provide a non-shortest reconfiguration sequence on bipartite permutation graphs.
We particularly emphasize that both reconfiguration problems remain PSPACE-complete even if two input independent sets are maximum independent sets of the input graph~\cite{BonsmaC:TCS:Finding:2009}.

As mentioned above, \prb{Independent Set Reconfiguration} can be solved in polynomial time on the class of even-hole free graphs.
In fact, for even-hole free graphs, Kami\'{n}ski et al.~\cite{KaminskiMM12} showed that every instance of \prb{Independent Set Reconfiguration} is a yes-instance (assuming that two input independent sets have the same cardinality).
This phenomenon does not hold on the class of permutation graphs: The instance consisting of $G \coloneqq K_{2, 2}$ with two color classes $I$ and $J$ is a no-instance, and $G$ is indeed a permutation graph, corresponding to sequence $A = (7,8,5,6)$.\footnote{For elements in $A$, we rather use integers more than $n$ to distinguish from their indices in some concrete examples.}
Also both $I = \{1, 2\}$ and $J = \{3, 4\}$ are maximum independent sets of $G$.
Thus, it is non-trivial to design a polynomial-time algorithm for \prb{Longest Increasing Subsequence Reconfiguration}.
Our polynomial-time algorithm exploits a structural property of the set of feasible sets for a given sequence $A$.

\section{Algorithm}

Let $A = (a_i)_{i = 1, 2, \ldots, n}$ be a sequence of $n$ distinct integers between $1$ and $n$.
Let $V = [n]$ and let $P = (V, \preceq_A)$ be a partial order on $V$ such that for $i, j \in V$,
\begin{align*}
    i \preceq_A j \iff (i = j) \lor (i < j \land a_i < a_j). 
\end{align*}
Then, a subset of $[n]$ is feasible for $A$ if and only if it is a chain of this partial order.
Moreover, by Mirsky's theorem~\cite{Mirsky71}, the largest size of a chain of $P$ is equal to the minimum size of an antichain partition of $V$, and such a partition can be computed in $O(n \log n)$ time by a standard dynamic programming algorithm for the longest increasing subsequence problem (see~\cite{Fredman:computing:DM:1975} for example).

To understand the structure of a minimum antichain partition, we use a specific construction, called \emph{patience sorting}~\cite{AldousD:Longest:BAMS:1999}, which is briefly described as follows.
For simplicity, we add $0$ to $A$ as $a_0 = 0$.
We use $n + 1$ piles $P_0, P_1, \ldots, P_n$ that are initially all empty and iteratively put an integer $a_i$ in $A$ on the top of one of the piles for $0 \le i \le n$ in this order.
For each $0 \le i \le n$, we put $a_i$ on the top of the ``leftmost'' pile $P_j$ such that $P_j$ is empty or the top of $P_j$ is greater than $a_i$.
Let us note that the top elements of all nonempty piles are always sorted in increasing order.
Now, let $P_0, P_1, \ldots, P_n$ be the piles obtained by executing the above algorithm for $A$.
Clearly, $P_0$ only contains $a_0$.
For each pile $P_i$, observe that $P_i$ is an antichain (with respect to $\preceq_A$): If $a_i$ is placed below $a_j$ in the pile, then $i < j$ and $a_i > a_j$.
For each $1 \le i \le n$, when $a_i$ is placed on the top of $P_k$ for some $1 \le k \le n$, the top element $a_j$ of $P_{k - 1}$ is smaller than $a_i$ (i.e., $a_j < a_i$).
In this case, we say that \emph{$a_j$ blocks $a_i$} and \emph{$a_i$ is blocked by $a_j$}.
Let $k$ be the largest index of a nonempty pile.
By definition, for $1 \le i \le n$, each $a_i$ has a unique element $a_j$ that blocks $a_i$.
Moreover, if $a_i$ is blocked by $a_j$, we have $a_j \preceq_A a_i$.
This implies that there is a chain (with respect to $\preceq_A$) of size $k + 1$, which corresponds to a feasible set $I$ for $A$.
As each $P_i$ is an antichain, this chain contains exactly one element of $P_i$ for each $0 \le i \le k$.
Thus, $I$ is a maximum feasible set for $A$.
The above construction of piles further implies the following observations.

\begin{observation}\label{obs:inc}
    Let $I$ be an arbitrary maximum feasible set for $A$.
    \begin{enumerate}
        \item If $a_v$ is placed below $a_u$ in a pile $P_i$, then we have $u > v$ and $a_u < a_v$.
        \item Each pile $P_i$ contains exactly one element $a_u$ with $u \in I$.
        \item Let $u, v \in I$ such that $a_u \in P_i$ and $a_v \in P_j$ for $0 \le i < j \le k$. Then, we have $u \preceq_A v$.
    \end{enumerate}
\end{observation}
\begin{proof}
    The first statement follows from the construction of $P_i$.
    The second statement follows from the fact that $P_i$ is an antichain with respect to $\preceq_A$.
    For the third statement, it suffices to show that $u < v$ (as the feasibility of $I$ implies that $a_u < a_v$).
    Suppose for contradiction that $u > v$.
    When $a_u$ is placed on the top of $P_i$, the top element on a pile $P_{j}$ is strictly larger than $a_u$.
    This and the first statement together imply that $a_u < a_v$, contradicting the fact that $I$ is feasible.
\end{proof}

Now, we turn to \prb{Longest Increasing Subsequence Reconfiguration}.
Let 
\begin{align*}
    \mathcal I = \{I \subseteq \{0\} \cup [n] : I \text{ is a maximum feasible set of } A\}    
\end{align*}
and let $P_0, P_1, \ldots, P_k$ be the nonempty piles that are obtained by applying the above algorithm to $A = (a_i)_{i = 0, 1, \ldots, n}$ with $a_0 = 0$.
The following observation follows from (2) in \Cref{obs:inc}.

\begin{observation}\label{obs:pile}
    Let $I, J \in \mathcal I$ such that $I \setminus J = \{u\}$ and $J \setminus I = \{v\}$.
    Then, $a_u, a_v \in P_j$ for some $0 \le j \le k$.
\end{observation}

Our algorithm for \prb{Longest Increasing Subsequence Reconfiguration} is based on a certain equivalence relation on $\mathcal I$.
For $I, J \in \mathcal I$,  we denote by $I \Next J$ if $I \setminus J = \{u\}$ and $J \setminus I = \{v\}$ such that $a_u$ is placed (strictly) below $a_v$ on pile $P_i$ for some $1 \le i \le k$.
We note that this $\Next$ relation is not transitive: $I \Next I'$ and $I' \Next I''$ may not imply $I \Next I''$.
For $I \in \mathcal I$, a family of feasible sets $\mathcal M(I) \subseteq \mathcal I$ is defined inductively as follows: (1) $M(I)$ contains $I$ and (2) for every $J \in \mathcal M(I)$, $J' \Next J$ implies $J' \in \mathcal M(I)$. 
In other words, $\mathcal M(I)$ is the lower set of $I$ in the transitive closure of $\Next$ in $\mathcal I$.
By definition, for $I \in \mathcal I$, $\mathcal M(J) \subsetneq \mathcal M(I)$ if $J \in \mathcal M(I)$ with $J \neq I$.
We say that $I \in \mathcal I$ is \emph{$\Next$-minimal} if there is no $J \in \mathcal I$ with $J \Next I$.

\begin{lemma}\label{lem:com}
    Let $I, J, J' \in \mathcal I$ such that $J \Next I$, $J' \Next I$, and $J \neq J'$.
    Then, at least one of the following conditions is satisfied: $J' \Next J$, $J \Next J'$, or there is $J'' \in \mathcal I$ such that $J'' \Next J$ and $J'' \Next J'$.
\end{lemma}
\begin{proof}
    Let $I \setminus J = \{u\}$, $J \setminus I = \{v\}$, $I \setminus J' = \{u'\}$, and $J' \setminus I = \{v'\}$.
    If $u$ and $u'$ belong to the same pile $P_i$, by~\Cref{obs:pile}, $v$ and $v'$ belong to the same pile $P_i$.
    This implies either $J' \Next J$ or $J \Next J'$.
    Suppose otherwise.
    By~\Cref{obs:pile}, $v$ and $v'$ belong to distinct piles and hence $v \neq v'$.
    We claim that $(J \setminus \{u'\}) \cup \{v'\}$ is a maximum feasible set, which symmetrically implies that $(J' \setminus \{u\}) \cup \{v\}$ is a maximum feasible set as well.
    Suppose for contradiction that $(J \setminus \{u'\}) \cup \{v'\}$ is not a feasible set.
    Since $J \setminus \{u'\}$ and $J' = (I \setminus \{u'\}) \cup \{v'\}$ are feasible, $v$ and $v'$ are the unique incomparable pair with respect to $\preceq_A$ in $(J \setminus \{u'\}) \cup \{v'\}$.
    We assume that $a_v$ and $a_{v'}$ are contained in piles $P_i$ and $P_j$ with $i < j$, respectively.
    As $v, u' \in J$, by~\Cref{obs:inc}, we have $a_{v} < a_{u'}$.
    Moreover, as $a_{v'}$ is placed below $a_{u'}$ in $P_j$, we have $a_{u'} < a_{v'}$ (by (1) in \Cref{obs:inc}).
    These together imply that $a_{v} < a_{v'}$.
    As $a_v$ is placed below $a_u$ in $P_i$, we have $v < u$ (by (1) in \Cref{obs:inc}).
    Moreover, by (3) in \Cref{obs:inc}, $u \preceq_A v'$ as $u, v' \in J'$.
    Thus, we have $v < v'$, contradicting the assumption that $v$ and $v'$ are incomparable with respect to $\preceq_A$.
\end{proof}

\begin{lemma}\label{lem:minimal}
    For $I \in \mathcal I$, there is exactly one $\Next$-minimal set in $\mathcal M(I)$.
\end{lemma}
\begin{proof}
    We prove the lemma by induction on $|\mathcal M(I)|$.
    If $|\mathcal M(I)| = 1$, then $I$ itself is the unique $\Next$-minimal set in $\mathcal M(I)$.
    Suppose that $\mathcal M(I)$ contains at least two sets.
    If there is exactly one $J \in \mathcal M(I)$ with $J \Next I$, by the induction hypothesis, $\mathcal M(J) \subsetneq \mathcal M(I)$ has a unique $\Next$-minimal set, which is also the unique $\Next$-minimal set in $\mathcal M(I)$.
    Otherwise, there are two $J, J' \in \mathcal M(I)$ such that $J \Next I$ and $J' \Next I$.
    By~\Cref{lem:com}, at least one of the following conditions are satisfied: $J' \Next J$, $J \Next J'$, or there is $J'' \in \mathcal I$ such that $J'' \Next J$ and $J'' \Next J'$.
    If $J' \Next J$, then $\mathcal M(J') \subseteq \mathcal M(J) \subsetneq \mathcal M(I)$.
    By induction, both $\mathcal M(J)$ and $\mathcal M(J')$ have unique $\Next$-minimal sets, and as $\mathcal M(J') \subseteq \mathcal M(J)$, these two sets are identical.
    The case where $J \Next J'$ is symmetric.
    Hence, suppose that there is $J'' \in \mathcal I$ such that $J'' \Next J$ and $J'' \Next J'$.
    By induction, $\mathcal M(J)$, $\mathcal M(J')$, and $\mathcal M(J'')$ have unique $\Next$-minimal sets.
    Similarly, as $\mathcal M(J'') \subseteq \mathcal M(J)$ and $\mathcal M(J'') \subseteq \mathcal M(J')$, these three $\Next$-minimal sets are identical, which completes the proof.
\end{proof}

The proof of \Cref{lem:minimal} immediately implies the following corollary.

\begin{corollary}\label{cor:minimal}
    For $I, J \in \mathcal I$ with $I \Next J$, the $\Next$-minimal sets of $\mathcal M(I)$ and $\mathcal M(J)$ are identical.
\end{corollary}

We define an equivalence relation on $\mathcal I$ based on the $\Next$-minimality.
By~\Cref{lem:minimal}, the $\Next$-minimal set in $\mathcal M(I)$ is uniquely determined for $I \in \mathcal I$.
We say that two maximum feasible sets $I$ and $J$ are \emph{$\Next$-equivalent} if the $\Next$-minimal set in $\mathcal M(I)$ is equal to that in $\mathcal M(J)$.
The key to our algorithm is the following lemma.

\begin{lemma}\label{lem:main}
    Let $I, J \in \mathcal I$.
    Then, there is a reconfiguration sequence between $I$ and $J$ if and only if $I$ and $J$ are $\Next$-equivalent.
\end{lemma}

\begin{proof}
    Suppose that there is a reconfiguration sequence $(I_0, I_1, \ldots, I_\ell)$ between $I_0 = I$ and $I_\ell = J$.
    We prove that all maximum feasible sets $I_i$ belong to the same $\Next$-equivalence class.
    By definition, either $I_{i} \Next I_{i + 1}$ or $I_{i + 1} \Next I_{i}$, implying respectively that $\mathcal M(I_i) \subseteq \mathcal M(I_{i + 1})$ or $\mathcal M(I_{i + 1}) \subseteq \mathcal M(I_{i})$.
    By~\Cref{cor:minimal}, their $\Next$-minimal sets are identical, which proves the forward direction.
    
    Suppose that $I$ and $J$ are $\Next$-equivalent.
    Then, there is $I' \in \mathcal M(I) \cap \mathcal M(J)$.
    This implies that there are reconfiguration sequences between $I$ and $I'$ and between $J$ and $I'$.
    By concatenating these sequences, we have a reconfiguration sequence between $I$ and $J$.
\end{proof}

Our algorithm is fairly straightforward.
Given two maximum feasible sets $I$ and $J$, we compute their $\Next$-minimal sets $I'$ and $J'$, respectively.
By~\Cref{lem:main}, there is a reconfiguration sequence between $I$ and $J$ if and only if $I' = J'$.
From a maximum feasible set $I$, we can compute a unique $\Next$-minimal set in $\mathcal M(I)$ in polynomial time by a greedy algorithm.
Hence, \Cref{thm:main} follows.

\section{Bipartite case}

Before proving \Cref{thm:bp}, we would like to mention that bipartiteness in \Cref{thm:bp} is crucial, that is, 
\prb{Longest Increasing Subsequence Reconfiguration} does not admit a reconfiguration sequence of length $|I \setminus J|$ in general.
Let us consider an instance consisting of $A = (15, 11, 16, 13, 17, 12, 14)$\footnote{Again, we use integers more than $n$ for the elements in $A$ to avoid confusion.}, $I = \{1,3,5\}$, and $J = \{2,6,7\}$.
This instance requires four steps to transform $I$ into $J$: $I_0 = \{1, 3, 5\} = I$, $I_1 = \{2,3,5\}$, $I_2 = \{2,4,5\}$, $I_3 = \{2,4,7\}$, $I_4 = \{2,6,7\} = J$, while $|I\setminus J| = 3$.

Let $(A = (a_i)_{i=1,2,\ldots, n}, I, J)$ be an instance of \prb{Longest Increasing Subsequence Reconfiguration} such that the underlying permutation graph $G_A$ of $A$ is bipartite.
In the following, we may not distinguish the elements of $A$ from their indices and then also refer to the elements of $A$ as the vertices of $G_A$.
Let $P_1, P_2, \ldots, P_k$ be the piles for $A$ defined in the previous section.
By (1) in \Cref{obs:pile}, every pair of indices of elements in a pile is incomparable with respect to $\preceq_A$.
This implies that they are adjacent in the permutation graph $G_A$.
Thus, each pile contains at most two elements as otherwise $G_A$ contains a triangle.
A pile $P_t$ is called a \emph{mixed pile} if it contains exactly two elements $a_{i}$ and $a_{j}$ with $i \in I$ and $j \in J$.
Note that, for such a mixed pile $P_t$, both $j \notin I$ and $i \notin J$ hold.
A pair of two mixed piles is called a \emph{forbidden pair} if the four vertices corresponding to two mixed piles induce a cycle of length $4$ in $G_A$.
It is easy to observe that $(A, I, J)$ is a no-instance if it has a forbidden pair.

A mixed pile $P_i$ is called the \emph{leftmost} mixed pile if no pile $P_j$ with $j < i$ is mixed.
The following lemma is a key to proving \Cref{thm:bp}.
\begin{lemma}\label{lem:mixed-pile}
    Suppose that $(A, I, J)$ has no forbidden pairs.
    Let $a_i, a_j$ be the elements in the leftmost mixed pile $P_{t}$ with $i \in I$ and $j \in J$.
    Then, at least one of $(I \setminus \{i\}) \cup \{j\}$ or $(J \setminus \{j\}) \cup \{i\}$ is feasible.
\end{lemma}
\begin{proof}
    Suppose that both $I' = (I \setminus \{i\}) \cup \{j\}$ and $J' = (J \setminus \{j\}) \cup \{i\}$ are not feasible.
    As $I'$ is not feasible, there is $i' \in I \setminus \{i\}$ that is adjacent to $j$ in $G_A$.
    Let $P_{t'}$ be the pile containing $a_{i'}$.
    Since $j \in J$, pile $P_{t'}$ has an element $a_{j'}$ with $j' \in J$, which implies that $P_{t'}$ is a mixed pile with $t < t'$.
    Symmetrically, as $J'$ is not feasible, there is a mixed pile $P_{t''}$ with $t < t''$ that has an element $a_{j''}$ with $j'' \in J \setminus \{j\}$ adjacent to $i$ in $G_A$.
    If $t' = t''$, the pair $P_t$ and $P_t'$ forms a forbidden pair, contradicting the assumption.
    Assume, without loss of generality, that $t < t' < t''$.
    Since there are edges between $j$ and $i'$ and between $i$ and $j''$, we have $a_j > a_{i'}$ and $a_i > a_{j''}$.
    As $j, j'' \in J$, we have $a_{j} < a_{j''}$.
    Thus, we have $a_{i'} < a_j < a_{j''} < a_i$, contradicting to the fact $a_{i} < a_{i'}$ as $i, i' \in I$.
\end{proof}

By~\Cref{lem:mixed-pile}, at least one of $(I \setminus \{i\}) \cup \{j\}$ or $(J \setminus \{j\}) \cup \{i\}$, say $I' = (I \setminus \{i\}) \cup \{j\}$, is feasible.
This decreases the difference $|I'\setminus J|$ by $1$ and does not create a new forbidden pair.
Applying repeatedly this, \Cref{thm:bp} follows.



\begin{sloppypar}
\printbibliography
\end{sloppypar}

\end{document}